\documentclass{ecai} 

\usepackage{latexsym}
\usepackage{amssymb}
\usepackage{amsmath}
\usepackage{amsthm}
\usepackage{booktabs}
\usepackage{graphicx}
\usepackage{color}
\usepackage{xcolor}
\usepackage{tcolorbox}
\usepackage{tikz}
\usepackage{enumerate}

\newtheorem{theorem}{Theorem}[section]
\newtheorem{lemma}[theorem]{Lemma}
\newtheorem{corollary}[theorem]{Corollary}
\newtheorem{proposition}[theorem]{Proposition}

\newtheorem{definition}{Definition}

\newcommand{\BibTeX}{B\kern-.05em{\sc i\kern-.025em b}\kern-.08em\TeX}
\newcommand{\util}{\textsc{Util}}
\newcommand{\egal}{\textsc{Egal}}
\newcommand{\egaldec}{\textsc{Egal-Dec}}
\newcommand{\nash}{\textsc{Nash}}
\newcommand{\nashdec}{\textsc{Nash-Dec}}
\newcommand{\mechanism}{\mathcal{M}}

\usepackage{todonotes}

\begin{document}

\begin{frontmatter}

\title{Temporal Elections:\\ Welfare, Strategyproofness, and Proportionality}

\author[1]{\fnms{Edith}~\snm{Elkind}}
\author[2]{\fnms{Tzeh Yuan}~\snm{Neoh}}
\author[3]{\fnms{Nicholas}~\snm{Teh}} 

\address[1]{Northwestern University, USA}
\address[2]{Agency for Science, Technology and Research, Singapore}
\address[3]{University of Oxford, UK}

\begin{abstract}
We investigate a model of sequential decision-making where a single alternative is chosen at each round. We focus on two objectives---utilitarian welfare (\util{}) and egalitarian  welfare (\egal{})---and consider the computational complexity of maximizing these objectives, as well as their compatibility with strategyproofness and proportionality. We observe that maximizing \util{} is easy, but the corresponding decision problem for \egal{} is NP-complete even in restricted cases. We complement this hardness result for \egal{} with parameterized complexity analysis and an approximation algorithm. Additionally, we show that, while a mechanism that outputs an outcome that maximizes \util{} is strategyproof, all deterministic mechanisms for computing outcomes that maximize \egal{} fail a very weak variant of strategyproofness, called non-obvious manipulability (NOM). However, we show that when agents have non-empty approval sets at each timestep, choosing an \egal{}-maximizing outcome while breaking ties lexicographically satisfies NOM. Regarding proportionality, we prove that a proportional (PROP) outcome can be computed efficiently, but finding an outcome that maximizes \util{} while guaranteeing PROP is NP-hard. We also derive upper and lower bounds on the (strong) price of proportionality with respect to \util{} and \egal{}. Some of our results extend to $p$-mean welfare measures other than \egal{} and \util{}.
\end{abstract}
\end{frontmatter}


\section{Introduction}
Consider a group of friends planning their itinerary for a two-week post-graduation trip across Europe. They have selected their activities, but still need to decide on their choice of meals for each day. As popular restaurants typically require reservations, everyone is asked to declare their preferences upfront before the trip commences.

Suppose that 55\% of the group members prefer Asian cuisine, 25\% prefer European cuisine, 10\% prefer Oceanic cuisine, and the remaining 10\% prefer South American cuisine.
If they need to schedule 40 meals, it would be fair to select European restaurants 
for 10 of these meals, and plan 22 visits to Asian restaurants, with the remaining 8 meals split equally between Oceanic and South American establishments.
However, if the friends were to adopt plurality voting to decide on each meal, 
Asian cuisine would be chosen for every meal, and, as a result, 45\% of the group will be perpetually unhappy. A natural question is then: what would be an appropriate notion of \emph{satisfaction}, and can we (efficiently) select an outcome that offers high satisfaction?

As the group moves from city to city, the set of available restaurants changes. Even within the same town, one may have different preferences for lunch and dinner, opting for a heavier meal option at lunch and a lighter one at dinner. As both preferences and the set of alternatives may evolve over time, traditional multiwinner voting models \cite{elkind2017propertiesmwv,faliszewski2017mwv,lackner2022abc} do not fully capture this setting.

This problem fits within the \emph{temporal elections} framework,
a model where a sequence of decisions is made, and  outcomes are evaluated with respect to agents' temporal preferences.
It was first introduced as \emph{perpetual voting} by \citet{lackner2020perpetual}; see the survey of \citet{elkind2024temporal} for a discussion of subsequent work.
%
We consider the \emph{offline} variant of this model where preferences are provided upfront.
That is, at each timestep, each agent has a set of approved alternatives, and the goal is to select a single alternative per timestep.

While this model has been considered in prior work~\cite{bulteau2021jrperpetual,chandak23,elkind2024verifying}, 
earlier papers focus on fairness and proportionality notions, with only a few (if at all) looking into welfare objectives and strategic considerations.
Against this background, in this work we focus on the task of maximizing two classic welfare objectives: the utilitarian welfare (the sum of agents' utilities) and
the egalitarian welfare (the minimum utility of any agent),  both in isolation and in combination with strategyproofness and proportionality axioms. We refer to these objectives (as well as, overloading notation, 
to outcomes that maximize them)
as \util{} and \egal{}, respectively.

\paragraph{Our Contributions}
In this paper, we investigate the \util{} and \egal{} objectives from three perspectives: the computational complexity of welfare maximization, compatibility with strategyproofness, and  trade-offs with proportionality.

In Section~3, we observe that a \util{} outcome can be computed in polynomial time, whereas computing an \egal{} outcome is NP-hard even under strong restrictions on the input. To mitigate the hardness result for \egal{} maximization, we analyze the parameterized complexity of this problem with  respect to several natural parameters, and provide an ILP-based approximation algorithm.

In Section~4, we show that, while a mechanism that outputs a \util{} outcome is strategyproof, any deterministic mechanism for \egal{} fails non-obvious manipulability (NOM), which is a relaxation of strategyproofness. On the positive side, if each agent has a non-empty approval set at each timestep, the mechanism that selects among \egal{} outcomes
using leximin tie-breaking satisfies NOM;
however, even in this special case
the \egal{} objective admits no deterministic strategyproof mechanism in general. 

In Section~5, we focus on proportionality. We show that, while a simple greedy algorithm can return a proportional (PROP) outcome, it is NP-hard to determine if there exists a PROP outcome that maximizes \util{}, even when each agent has a non-empty approval set at each timestep.
We also provide upper and lower bounds on the (strong) price of proportionality with respect to both \util{} and \egal{}. 
To the best of our knowledge, our work is the first to investigate the price of proportionality in temporal elections. 
Towards the end of the paper, we discuss how to extend our results to the general class of $p$-mean welfare objectives.

\paragraph{Related Work}
\citet{lackner2020perpetual} and, subsequently,  \citet{lackner2023proportionalPV} put forward the framework of {\em perpetual voting} in order to capture fairness in repeated decision-making.
They focused on temporal extensions of traditional multiwinner voting rules, and analyzed them with respect to novel axioms developed for the temporal setting.
\citet{bulteau2021jrperpetual} built upon this framework and proposed temporal extensions of popular \emph{justified representation} axioms.
\citet{chandak23} continued this line of work by studying whether temporal variants of justified representation axioms can be satisfied by efficiently computable voting rules. Their analysis is complemented by that of
\citet{elkind2024verifying}, who considered the computational problems associated with verifying whether outcomes satisfy these axioms.

An adjacent line of work by \citet{bredereck2020successivecommittee,bredereck2022committeechange} and \citet{zech2024multiwinnerchange}
looks into sequential committee elections, where an entire committee (i.e., a fixed-size set of candidates) is elected at each timestep. It is assumed that voters' preferences may evolve, but it may be desirable to minimize changes in the composition of the committee, or, conversely, to replace as many committee members as possible; the goal is to select a good committee (according to a given voting rule) under constraints of this form.
\citet{deltl2023seqcommittee} consider a similar model, but with the restriction that agents can only approve at most one candidate per timestep.

Our work is also related to apportionment with approval preferences
\cite{brill2024partyapportionment,delemazure2022spelection}.
In this setting, the goal is to allocate the seats of a fixed-size committee to parties based on voters' (approval) preferences over the parties.
This is equivalent to a restricted setting of temporal voting where voters have preferences that are static, i.e., do not change over time.

Yet another related model is that of fair scheduling: just as in our setting, voters are assumed to have approval preferences over the candidates at each timestep, but the key difference is that each candidate can be selected at most once.
\citet{elkind2022temporalslot}  studied computational issues associated with maximizing various welfare objectives in this model, along with several fairness properties.

One could also view temporal elections as allocating public goods or decision-making on public issues, making the fair public decision-making model \cite{conitzer2017fairpublic,fain2018publicgoods,skowron2022proppublic} particularly relevant. 

There are other models of issue-by-issue decision-making that are similar in spirit to temporal voting.
In particular, \citet{alouf2022better} considered issue-by-issue voting with uncertainty in voters' preferences, where the goal is to recover the majority-supported outcome for each issue.

Our analysis of price of proportionality is inspired by prior work on similar phenomena (i.e., the loss of welfare caused by imposing fairness guarantees) in the context of single-round multiwinner voting \cite{brill2024priceable,elkind2022price,lackner2020util}.


\section{Preliminaries}\label{sec:prelim}
For each positive integer $k$, let $[k] := \{1,\dots,k\}$.
Let $N = [n]$ be a set of $n$ \emph{agents}, let $P = \{p_1,\dots,p_m\}$ be a set of $m$ \emph{projects} (or \emph{candidates}), and let $\ell$ denote the number of \emph{timesteps}.
For each $t \in [\ell]$, the set of projects approved by agent $i\in N$ at timestep $t$
is captured by her \emph{approval set} $S_{it} \subseteq P$. The approval sets of agent $i$
are collected in her \emph{approval vector}
$\mathbf{S}_i = (S_{i1}, \dots, S_{i\ell})$.
Thus, an {\em instance} of our problem is a tuple 
$(N, P, \ell, ({\mathbf S}_i)_{i\in N})$. 
For each $p\in P, t\in [\ell]$ we denote by $n_{pt}$ the number of agents in $N$ that approve project $p$ at timestep $t$.

An \emph{outcome} is a vector $\mathbf{o} = (o_1,\dots,o_\ell)$, where $o_t\in P$ for each $t \in [\ell]$.
The utility of an agent $i \in N$ for an outcome 
$\mathbf{o}$ is given by $u_i(\mathbf{o}) = |\{t \in [\ell]: o_t \in S_{it}\}|$. Let $\Pi({\mathcal I})$ denote the space of all possible outcomes for an instance $\mathcal I$.
A {\em mechanism} maps an instance ${\mathcal I} = (N, P, \ell, ({\mathbf S}_i)_{i\in N})$ 
to an outcome in $\Pi({\mathcal I})$. 

We do not require each agent to approve at least one project at each timestep;
however, we do require that each agent approves at least one project at some timestep, i.e., for each $i\in N$ there exists a 
$t\in [\ell]$ with $S_{it} \neq \varnothing$; indeed, if this condition is failed for some $i\in N$, we can simply delete $i$, as it can never be satisfied. 
If $S_{it} \neq \varnothing$ for all $i\in N$ and $t \in [\ell]$, we say that the agents have {\em complete preferences (CP)}. 
We believe that the CP setting captures many real-life applications of our model:
for instance, in our motivating example, having no particular opinion on any cuisine would be more reasonably interpreted as approving all options rather than having an empty approval set.
We also consider the more restricted setting where $|S_{it}|=1$ for all $i\in N$ and $t \in [\ell]$; in this case, we say that the agents have {\em unique preferences (UP)}. 

We assume that the reader is familiar with basic notions of classic complexity theory~\cite{papadimitriou_computational_2007} 
and parameterized complexity \cite{flum_parameterized_2006,niedermeier_invitation_2006}.


\section{Welfare Maximization}\label{sec:welfare}
We first focus on welfare maximization; subsequently, we will explore combining welfare objectives with other desiderata.
The two objectives we consider are defined as follows.

\begin{definition}[Social Welfare] \label{defn-util} 
    Given an outcome $\mathbf{o}$, its
    \emph{utilitarian social welfare} 
    is defined as $\util({\mathbf o})=\sum_{i \in N} u_i(\mathbf{o})$ and 
    its
    \emph{egalitarian social welfare} 
    is defined as $\egal({\mathbf o}) = \min_{i \in N} u_i(\mathbf{o})$.
    We refer to outcomes that maximize the utilitarian/egalitarian welfare
    as \util{}/\egal{} outcomes, respectively.
\end{definition}


A \util{} outcome can be found
in polynomial time: at each timestep, we select a project that has the highest number of approvals.
However, if 51\% of the population approves
project $p$ (and nothing else) at each timestep, while 49\% of the population approves project $q$ (and nothing else) at each timestep, the 
outcome $(p, \dots, p)$ is the unique \util{} outcome, but close to half of the  population will obtain utility $0$ from it.
Thus, it is important to consider desiderata other than \util{}, such as, e.g., egalitarian welfare.
However, while computing a \util{} outcome  is computationally feasible, this is not the case for \egal{} outcomes.

The decision problem associated with computing \egal{} outcomes, which we denote by \egaldec{}, 
is defined as follows.
\begin{tcolorbox}
\textsc{Maximizing Egalitarian Welfare (Egal-Dec)}:\\
\textbf{Input}: An instance ${\mathcal I} = (N,P,\ell, (\mathbf{S}_i)_{i\in N})$ 
and a parameter $\lambda \in \mathbb{Z}^+$.\\
\textbf{Question}: Is there an outcome $\mathbf{o}$ with  
$\util{}(\mathbf{o})\geq \lambda$?
\end{tcolorbox}

The following result shows that \egaldec{} is NP-complete, even in the UP setting and even when the goal is to guarantee each agent a utility of $1$.
\begin{theorem} \label{thm:npcomplete}
    \egaldec{} is \emph{NP}-complete, even in  
    the UP setting with $\lambda = 1$.
\end{theorem}
\begin{proof}
 It is immediate that \egaldec{} is in NP: given a candidate outcome, we can easily check
 if it provides utility of at least $\lambda$ to each agent.

 To prove hardness, we reduce from the classic {\sc Vertex Cover} problem. An instance of this problem
 is a pair $(G, k)$, where $G=(V, E)$ is an undirected graph 
 and $k$ is a positive integer; it is a yes-instance if there is a subset $V'\subseteq V$
 of size at most $k$ such that for every edge $\{v, v'\}\in E$ we have $v\in V'$ or $v'\in V'$, and a no-instance otherwise. 

 Given an instance $(G, k)$ of {\sc Vertex Cover} with $G=(V, E)$, $V=\{v_1, \dots, v_s\}$, $|E|=r$, we proceed as follows. 
 We let $N=[r+s-k]$, where the first $r$ agents in $N$ are edge agents (i.e., one agent for each edge of $G$), 
 and the remaining $s-k$ agents are dummy agents (where $k$ is the target size of the vertex cover, so $k\le s$), 
 and set $n=|N|=r+s-k$, $\ell = s$.

 Let $P^0=\{q_1, \dots, q_n\}$, $P^*=\{p_1, \dots, p_s\}$, and set $P=P^0\cup P^*$.
 For each $t\in [\ell]$, we set the approval set of the edge agent $i$ at timestep $t$ to $\{p_t\}$ if $i$ is an edge agent that corresponds to an edge incident to $v_t$; if $i$ is a dummy agent or an edge agent who corresponds to an edge not incident to $v_t$, we set $S_{it}=\{q_i\}$. By construction, the agents have unique preferences. Let $\lambda=1$.

 Suppose $G$ admits a vertex cover $V'$ of size $k$. We then construct an outcome 
 ${\mathbf o}=(o_1, \dots, o_\ell)$ as follows. For each $v_j\in V'$ we set
 $o_j=p_j$. The remaining $s-k$ timesteps are assigned arbitrarily to the $s-k$ dummy agents; if a dummy agent $i$ receives timestep $t$, we set $o_t=q_i$.
 By construction, the utility of every dummy agent is $1$. We claim that all edge
 agents have positive utility as well. Indeed, consider an agent that corresponds
 to an edge $\{v_j, v_t\}$. Since $V'$ is a vertex cover, we have $o_j=p_j$ or $o_t=p_t$ (or both), so the utility of
 this agent is at least $1$.

 Conversely, consider an outcome $\mathbf o$ with $u_i({\mathbf o})\ge 1$ for all $i\in N$. A dummy agent $i$ 
 can only be satisfied if we set $o_t=q_i$ for some
 $t\in [\ell]$, and this provides zero utility to agents other than~$i$. Thus, $\mathbf o$ has
 to allocate at least $s-k$ timesteps to satisfying the dummy agents. Hence, $\mathbf o$ uses at most
 $k$ timesteps to provide utility of $1$ to each edge agent. We will use this fact to construct 
 a set of vertices $V'$ that forms a vertex cover of size at most $k$ for $G$. 
 Let $T$ be the set of timesteps in $[\ell]$
 that are not allocated to dummy agents; we have $|T|\le k$. 
 We initialize $V'=\varnothing$. Then, 
 for each timestep $t\in T$, if $o_t=p_j$ for some $p_j\in P^*$ we place $v_j$ in $V'$, and if $o_t=q_j$ for some $j\in [r]$, we place one of the endpoints of the edge that corresponds to agent $j$ in $V'$.
 The resulting set $V'$ has size at most $|T|\le k$;
 moreover, as it `covers' each edge agent, it has to form a vertex cover for $G$. 
\end{proof}

Theorem~\ref{thm:npcomplete} effectively rules out the possibility of maximizing the egalitarian welfare even in simple settings.
Therefore, it what follows, we consider the complexity of egalitarian welfare maximization from the parameterized complexity perspective (Section~\ref{sec:param}) and from the approximation algorithms perspective (Section~\ref{sec:approx}).


\subsection{Parameterized Complexity of Egalitarian Welfare}\label{sec:param}
The \egaldec{} problem has four natural parameters: the number of agents $n$, the number of projects $m$, the number of timesteps $\ell$, and the utility guarantee $\lambda$. We will now consider these parameters one by one.



We first show that \egaldec{} is in FPT with respect to $n$. Our proof is based on integer linear programming. Specifically, we show how to encode \egaldec{} as an integer linear program (ILP) whose number of variables depends on $n$ only; our claim then follows from Lenstra's classic result~\cite{lenstra1983integer}. 
To accomplish this, we classify the projects and timesteps into `types', so that the number of types is (doubly) exponential in $n$, but does not depend on $m$ or $\ell$.
\begin{theorem} \label{thm:fpt}
    \egaldec{} is \emph{FPT} with respect to $n$.
\end{theorem}
\begin{proof}
    As a preprocessing step, we create $\ell$ copies of each project. That is, we replace a project
    $p$ with projects $p^1, \dots, p^\ell$ and modify the approval vectors: for each $i\in N$, 
    $t\in [\ell]$ and $p\in S_{it}$, we place $p^t$ in $S_{it}$ and remove $p$. This does not change the nature
    of our problem, since in our setting there is no dependence between timesteps, and a project's label can be re-used arbitrarily between timesteps.
    For the modified instance, 
    it holds that for each project $p$ there is at most one timestep $t\in [\ell]$
    such that $p\in\cup_{i\in N}S_{it}$. Then, we define the {\em type} of a project as the set of voters who approve it: the type of $p$ is $\pi(p) = \{i\in N: p\in \cup_{t\in [\ell]}S_{it}\}$. Because of the preprocessing step, for each project $p$
    there is a timestep $t\in [\ell]$ such that $p\in S_{it}$ for all $i\in \pi(p)$, and $p\not\in S_{it'}$ for all $i\in N$, $t'\in [\ell]\setminus\{t\}$. Let ${\mathcal P}\subseteq 2^N$ be the set of all project types;
    by construction, we have $|{\mathcal P}|\le 2^n$. 
    
    In the same way, we can classify timesteps by the types of projects present in them: the type of a timestep $t$ is defined as $\tau(t) = \{\pi(p): 
    p\in \cup_{i\in N}S_{it} \}$. Let $\mathcal{T} \subseteq 2^{\mathcal P}$ be the set of all timestep types; by construction, we have
    $|\mathcal{T}|\le 2^{|\mathcal{P}|}\le 2^{2^n}$.
    
    We are now ready to construct our ILP.
    For each $\tau \in \mathcal{T}$, let $z_\tau$ be the number of
    timesteps of type $\tau$; these quantities can be computed from the input.
    For each $\tau \in \mathcal{T}$ and $\pi\in \tau$, we introduce an integer variable $x_{\tau, \pi}$ representing the number of timesteps of type $\tau$ in which a project of type $\pi$ was chosen.

    The ILP is defined as follows: 
	\begin{equation*}
		\text{maximize } \lambda \quad\text{ subject to: }
	\end{equation*}
 \begin{itemize}
 \item[(1)] $\sum_{\pi \in \tau} x_{\tau, \pi} \leq z_\tau \text{ for each } \tau \in \mathcal{T}$;
 \item[(2)] $\sum_{\tau \in \mathcal{T}}\sum_{\pi \in {\mathcal P}: i\in \pi} x_{\tau, \pi} \geq \lambda \text{  for each } i \in N$;
 \item[(3)] $x_{\tau, \pi} \geq 0$ for each $\tau \in \mathcal{T}$ and $\pi\in \mathcal P$.
 \end{itemize}       
    \noindent
    The first constraint requires that we select at most one project per timestep, whereas the second constraint ensures that each agent has utility of at least $\lambda$ from the outcome.
    
    There are at most $\mathcal{O}(2^{n + 2^n})$ variables in the ILP, so the classic result of \citet{lenstra1983integer} implies that our problem is FPT with respect to $n$.
\end{proof}

For the number of agents $m$, we can show that \egaldec{} is NP-hard even for $m=2$; this hardness result holds even if $\lambda=1$, but our reduction produces instances where $S_{it}=\varnothing$ for some $i\in N$, $t\in [\ell]$, i.e., it does {\em not} show that \egaldec{} is NP-hard in the CP setting
(and hence Theorem~\ref{thm:m=2} below does not imply Theorem~\ref{thm:npcomplete}).
Our argument is similar in spirit to the proof of Theorem~2 of \citet{deltl2023seqcommittee}.

\begin{theorem}\label{thm:m=2}
    \egaldec{} is {\em NP}-complete even if $m=2$, $\lambda=1$.
\end{theorem}
\begin{proof}
    We reduce from the classical \textsc{3-SAT} problem. 
    Consider a formula $F$ with $\alpha$ variables $x_1,\dots,x_\alpha$, and $\kappa$ clauses $M_1,\dots,M_\kappa$.
    Each clause is a disjunction of at most three variables or their negations. 
    $F$ is satisfiable if there exists some assignment of Boolean values to variables such that the conjunction of all $\kappa$ clauses evaluates to True, and not satisfiable otherwise.
    
    To reduce \textsc{3-SAT} to our problem, we introduce $\kappa$ agents, 
    a set of projects $P=\{\tau, \varphi\}$, 
    and $\alpha$ timesteps; we set $\lambda=1$. 
    For each $i \in [\kappa]$ and $t \in [\alpha]$, we define $S_{it}$ as follows.
\begin{equation*}
S_{it} =
    \begin{cases}
        \{ \tau \}    & \text{if  } \,x_t \text{ is in }  M_i\\
        \{ \varphi \} & \text{if  } \,\neg x_t \text{ is in }  M_i\\
        \varnothing   & \text{otherwise }
    \end{cases}
\end{equation*}
We claim that $F$ is satisfiable if and only if there exists an outcome $\mathbf{o}$ that satisfies $u_i(\mathbf{o}) \geq 1$ for all $i \in [\kappa]$. 

For the `if' direction, an outcome $\mathbf{o} =(o_1, \dots, o_\alpha)$ can be interpreted a Boolean assignment 
to the variables in $F$: for each $t\in [\alpha]$ we set $x_t$ to True if $o_t=\tau$ and to False if $o_t=\varphi$.
Then for each $i \in [\kappa]$ her utility $u_i(\mathbf{o})$ is the number of literals in $M_i$ that are set to 
True by this assignment. Hence, if for all $i \in [\kappa]$ it holds that $u_i(\mathbf{o}) \geq 1$, 
then each clause will evaluate to True, and hence $F$ is satisfiable.

For the `only if' direction, given a satisfying Boolean assignment for $F$, 
we construct an outcome $ \mathbf{o} = (o_1, \dots, o_\alpha)$ by setting $o_t=\tau$ 
if this assignment sets $x_t$ to True and $o_t=\varphi$ otherwise.
We then have $u_i(\mathbf{o}) \geq 1$ for all $i \in [\kappa]$. 
\end{proof}

Next, we show that \egaldec{} is in XP with respect to $\ell$. Our argument is based on exhaustive search.

\begin{theorem} \label{thm:xp}
    \egaldec{} is in \emph{XP} with respect to $\ell$.
\end{theorem}
\begin{proof}
    Observe that there are $m^\ell$ possible outcomes. 
    We can thus iterate through all outcomes; we output `yes' if we find an outcome that provides utility $\lambda$ to all agents, and 'no' otherwise.
\end{proof}
Both Theorem~\ref{thm:fpt} and Theorem~\ref{thm:xp} provide algorithms for \egaldec{}, which is a decision problem. However, both algorithms can be modified to output an \egal{} outcome.

We complement our XP result for $\ell$ by showing that \egaldec{} is W[2]-hard with respect to the number of timesteps. This indicates that an FPT (in $\ell$) algorithm does not exist unless FPT = W[2], and hence the XP result of Theorem~\ref{thm:xp} is tight.

\begin{theorem} \label{thm:w2hard}
    \egaldec{} is \emph{W[2]}-hard with respect to $\ell$, even in the CP setting with $\lambda=1$. 
\end{theorem}
\begin{proof}
    We reduce from the \textsc{Dominating Set (DS)} problem.
An instance of \textsc{DS} consists of a graph $G = (V,E)$ and an integer $\kappa$; it is a yes-instance if there exists a subset $D \subseteq V$ such that $|D| \leq \kappa$ and every vertex of $G$ is either in $D$ or has a neighbor in $D$, and a no-instance otherwise.
\textsc{DS} is known to be W[2]-complete with respect to the parameter $\kappa$ \cite{niedermeier_invitation_2006}.

Given an instance $(G, \kappa)$ of \textsc{Dominating Set} with $G = (V,E)$, $V=\{v_1, \dots, v_n\}$,  
set $N=[n]$, $P = \{p_1,\dots,p_n\}$, $\ell = \kappa$.
Then for each $i\in N$
and $t\in [\ell]$
let $S_{it}=\{p_j: i=j\text{ or }\{v_i, v_j\}\in E\}$.
We claim that $G$ admits a dominating set $D$ with $|D| \leq \kappa$ if and only if there exists an outcome $\mathbf{o}$ such that $u_i(\mathbf{o}) \geq 1$ for all agents $i \in N$. 

For the `if' direction, consider an outcome $\mathbf{o} = (p_{j_1},\dots,p_{j_\kappa})$
that provides positive 
utility to all agents, and 
set $D=\{v_{j_1}, \dots, v_{j_\kappa}\}$. Then $D$ is a dominating set of size at most $\kappa$. 
Indeed, consider a vertex $v_i\in V$. Since voter $i$ approves $p_{j_t}$ for some $j_t\in T$, we have $v_{j_t}\in D$ 
and $i=j_t$ or $\{v_i, v_{j_t}\}\in E$.
Note that if there are projects chosen more than once, we simply have $|D|<\kappa$. 

For the `only if' direction, observe that a dominating set $D=\{v_{j_1}, \dots, v_{j_s}\}$ with $s \leq \kappa$ can be mapped to an outcome $\mathbf{o} = (p_{j_1}, \dots, p_{j_s}, p_1, \dots, p_1)$ (with $p_1$ selected in the last $\kappa-s$ timesteps). As any vertex of $G$ is either in $D$, or has a neighbor in $D$,  we have $u_i(\mathbf{o}) \geq 1$ for each agent $i \in N$.
\end{proof}

Finally, we consider parameter $\lambda$. By Theorems~\ref{thm:npcomplete} and~\ref{thm:m=2}, \egaldec{} is hard 
even for $\lambda=1$; this result holds even in
the UP setting (Theorem~\ref{thm:npcomplete}) or if $m=2$ (Theorem~\ref{thm:m=2}).
However, if both all approval sets are non-empty 
(which is a weaker condition than UP) and $m=2$, we obtain a positive result.

\begin{theorem} \label{thm:XP-lambda}
    \egaldec{} is in \emph{XP} with respect to $\lambda$ in the \emph{CP} setting with $m = 2$.
\end{theorem}
\begin{proof}
    Suppose first that $\ell \leq \lambda\cdot \lceil\log_2 n\rceil\leq \lambda\cdot(\log_2 n+1)$.
    We can then enumerate all $m^\ell\le (2n)^\lambda$
    possible outcomes, compute the egalitarian welfare for each outcome, and return `yes' if there exists an outcome in which the utility of each agent is at least $\lambda$. Therefore, from now we will assume that  
    $\ell > \lambda\cdot\lceil\log_2 n\rceil$. We will now argue that in this case we can greedily construct an outcome which guarantees utility of at least $\lambda$ to each agent, i.e., our instance 
    is a yes-instance of \egaldec{}.
    
    We split $[\ell]$ into $\lambda$ consecutive blocks $T_1, \dots, T_\lambda$ of length at least $\lceil\log_2 n\rceil$ each. It suffices to argue that for each block $T_j$ we can assign projects to timesteps in $T_j$ so that each agent derives positive utility from at least one timestep in $T_j$.
    
    Consider the block $T_1$. We start by setting $N^*=N$, and proceed in $|T_1|$ steps. Let $P=\{p, q\}$.
    At each step $t$, $t=1, \dots, |T_1|$, if at least half of the agents in $N^*$ approve $p$ at timestep $t\in T_1$, we set $o_t=p$ and remove all agents $i$ with $p\in S_{it}$ from $N^*$; otherwise, 
    we set $o_t=q$ and remove all agents $i$ with $q\in S_{it}$ from $N^*$. Note that an agent is removed from $N^*$ only after we ensure that she derives positive utility from at least one timestep in $T_1$. Moreover, since we are in the CP setting, at each timestep $t$
    at least one project in $P$ is supported by at least half of the remaining agents, so at each step we reduce the size of $N^*$ by at least a factor of $2$. It follows that after $\lceil\log_2 n\rceil\le |T_1|$ steps the set $N^*$ is empty, i.e., each agent in $N$ derives positive utility from some timestep in $T_1$; we can then assign projects to remaining timesteps in $T_1$ arbitrarily.
    By repeating this procedure for $T_2, \dots, T_\lambda$, we construct an outcome $\mathbf o$ with $u_i({\mathbf o})\ge \lambda$ for each $i\in N$; moreover, our procedure runs in polynomial time.
\end{proof}


\subsection{Approximation of Egalitarian Welfare}\label{sec:approx}
Theorem~\ref{thm:npcomplete} shows that \egaldec{} is NP-complete even when $\lambda = 1$.
This implies that the problem of computing the \egal{} welfare is inapproximable: an approximation algorithm would be 
able to detect whether a given instance admits an outcome with positive egalitarian social welfare.
However, suppose we redefine each agent's utility function as $u'_i(\mathbf{o})=1+u_i(\mathbf{o})$;
this captures, e.g., settings where there is a timestep in which all agents approve the same project.
We will now show that we can obtain a $\frac{1} {4\log n}$-approximation 
to the optimal egalitarian welfare with respect to the utility profile $(u'_1, \dots, u'_n)$.

\begin{theorem} \label{thm:approx}
    There is a polynomial-time randomized algorithm that, for any $\varepsilon>0$, 
    given an instance $(N, P, \ell, (\mathbf{S}_i)_{i \in N})$, with probability $1-\varepsilon$
    outputs an outcome $\mathbf o$ whose egalitarian social welfare is at least $\frac{1}{4\ln n}$
    of the optimal egalitarian social welfare with respect 
    to the modified utility functions $(u'_1, \dots, u'_n)$.
\end{theorem}
\begin{proof}
    First, we construct a polynomial-size integer program for finding outcomes
    whose egalitarian welfare with respect to the modified utilities is at least a given quantity $\eta$.
    For each $p \in P$ and $t \in [\ell]$, we define a variable $x_{pt}\in \{0,1\}$: 
    $x_{pt}=1$ encodes that $p$ is selected at time $t$. 
    Our constraints require that (1) for each $t\in [\ell]$, at most one project can be chosen at timestep $t$: 
    $\sum\limits_{p\in P} x_{pt}\leq 1$, and 
    (2) each agent $i\in N$ approves the outcome 
    in at least $\eta-1$ timesteps, so her modified utility is at least $\eta$:
    $\sum\limits_{t=1}^\ell \sum\limits_{p \in S_{it}} x_{pt} + 1 \geq \eta$. 
    We then maximize $\eta$ subject to these constraints.
    By relaxing the $0$-$1$ variables $x_{pt}$ to take values in $\mathbb{R}_+$, 
    we obtain the following LP relaxation:

\begin{align*}
    \textrm{LP-\egal{}}: 
    &\max \eta \\
    &\sum\limits_{p\in P} x_{pt}\leq 1 &\text{ for all } t\in [\ell]\\
    \quad &\sum\limits_{t=1}^\ell \sum\limits_{p \in S_{it}} x_{pt} \geq \eta- 1 
     &\text{ for all } i\in N\\
    &x_{pt}\geq 0 &\text{ for all } p\in P, t\in [\ell].
\end{align*}
Let $\eta^*$ be the optimal (fractional) value of LP-\egal{}, and let $(\{x^*_{pt}\}_{p\in P, t\in[\ell]}, \eta^*)$ be the associated fractional solution.
Let $\eta'$ be the optimal egalitarian
welfare with respect to $u'_1, \dots, u'_n$ in our instance. Then $\eta'$ together with an encoding
of the outcome that provides this welfare forms a feasible solution to LP-\egal{},
and hence $\eta'\le \eta^*$.

For every outcome $\mathbf o$ we have 
$u'_i(\mathbf{o})=u_i(\mathbf{o})+1\geq 1$ for all $i\in N$. Hence, if $\eta^*\leq 4 \ln n$, it holds that every outcome
$\mathbf o$ 
is a $\frac{1} {4\ln n}$-approximation. Thus, we can output an arbitrary outcome 
in this case. Therefore, from now we assume that $\eta^*> 4 \ln n$.

To transform $\{x^*_{pt}\}_{p\in P, t\in[\ell]}$ into a feasible integer solution, 
we use randomized rounding: for each $t\in [\ell]$
we set $o_t=p$ with probability $x^*_{pt}$. 
These choices are independent across timesteps.
For each $i\in N$ and $t\in [\ell]$, define a Bernoulli random variable $Z_i^t$ that
indicates whether agent $i$ approves the project randomly selected at timestep $t$. 
Then, for each agent $i\in N$ let $Z_i = \sum_{t=1}^\ell Z_i^t$, so that 
the utility of agent $i$ is given by 
$u'_i = Z_i + 1$. 
Then, the expected value of $Z_i^t$ is 
\begin{align*}
    \mathbb{E}[Z_i^t] = \sum_{p\in S_{it}} x^*_{pt}.
\end{align*}
By linearity of expectation, we derive 
\begin{align*}
    \mathbb{E}[Z_i] = \sum_{t=1}^\ell \mathbb{E}[Z_i^t] &= \sum_{t=1}^\ell 
    \sum_{p\in S_{it}} x^*_{pt}
    \geq \eta^*-1.
\end{align*}
 Applying the multiplicative Chernoff bound \cite{alonspencer}, we obtain
$$
    \mathbb{P}\{u'_i\leq \eta^*(1-\delta)\} \leq \text{exp}\left(\frac{-\eta^*\delta^2}{2}\right)\text{ for every $\delta>0$}.
$$
Recall that $\eta^*>4 \ln n$.
Thus, by letting $\delta = \frac{4}{5}$, we have 
\begin{align*}
    \mathbb{P}\left\{u'_i\leq \frac{\eta^*}{5}\right\} &\leq \text{exp}\left(-\frac{32\ln n}{25}\right) = n^{-\frac{32}{25}}.
\end{align*}
Finally, by applying the union bound, we get
    \begin{align*}
        \mathbb{P}\left\{u'_i\geq \frac{\eta^*}{5} \text{for all } i\in N\right\} &\geq  1- n \cdot n^{-\frac{32}{25}}
        &= 1- n^{-\frac{7}{25}}> 0.
    \end{align*}
Consequently, with positive probability we obtain an integer solution whose egalitarian welfare (with respect to the modified utilities) is at least $\frac15$ of the optimal egalitarian welfare. Using probability amplification
techniques, we can find some such solution with probability $1-\varepsilon$.
It remains to observe that $\frac{1}{5} > \frac{1} {4\ln n}$ when $n > 3$. 
\end{proof}


\section{Strategyproofness and Non-Obvious Manipulability}\label{sec:nom}
An important consideration in the context of collective decision-making is \emph{strategyproofness}: 
no agent should be able to increase their utility by misreporting their preferences. It is formally defined as follows. Note that agent $i$'s utility function $u_i$ is computed with respect to her (truthful) approval vector $\mathbf{S}_i$.

\begin{definition}[Strategyproofness] \label{defn-SP}
For each $i\in N$, let $\mathcal{S}_{-i}$ denote the list of all approval vectors except that of agent $i$: 
$\mathcal{S}_{-i} = (\mathbf{S}_1,\dots,\mathbf{S}_{i-1},\mathbf{S}_{i+1},\dots,\mathbf{S}_n)$.
    A mechanism $\mechanism$ is \emph{strategyproof} (SP) if for each instance $(N, P, \ell, ({\mathbf S}_i)_{i\in N})$, each
    agent $i \in N$ and each approval vector $\mathbf{B}_i$
    it holds that 
    $u_i(\mechanism(\mathcal{S}_{-i}), \mathbf{S}_i) \geq u_i(\mechanism(\mathcal{S}_{-i}), \mathbf{B}_i)$.
    \end{definition}

Consider the mechanism that outputs a \util{} outcome 
by choosing a project with the highest number of approvals at each timestep 
(breaking ties lexicographically); we will refer to this mechanism as \textsc{GreedyUtil}.
We observe that this mechanism is strategyproof. 

\begin{theorem}\label{thm:util_sp}
    \textsc{GreedyUtil} is strategyproof.
\end{theorem}
\begin{proof}
    Consider an agent $i \in N$ and a timestep $t\in [\ell]$. 
    Since \textsc{GreedyUtil} makes a decision for each timestep independently of agents' reports regarding other timesteps, it suffices to argue that $i$ cannot increase her utility at timestep $t$ by reporting an approval set 
    $S\neq S_{it}$. This follows directly from the fact that Approval Voting
    is strategyproof in the single-winner setting (see, e.g., \cite{BF83}); for completeness, we provide  a direct proof.

    Suppose that when $i$ truthfully reports $S_{it}$ at timestep $t$, {\sc GreedyUtil} set $o_t=p$. If $p\in S_{it}$, agent $i$ cannot benefit from misreporting at timestep $t$, so assume that 
    $p \notin S_{it}$. Then for every project $q\in S_{it}$ it holds that either $q$ gets fewer approvals
    than $p$ at timestep $t$, or $p$ and $q$ receive the same number of approvals, but $p$ precedes $q$
    in the tie-breaking order. Agent $i$ cannot change this by modifying her report, so she cannot increase
    her utility at timestep $t$. As this holds for every $t\in [\ell]$, the proof is complete.
\end{proof}
In contrast, no deterministic mechanism that maximizes egalitarian welfare can be strategyproof. Intuitively, 
this is because agents have an incentive to not report their approval for projects that are popular among other agents.

\begin{proposition}\label{prop:sp_no}
    Let $\mechanism$ be a deterministic mechanism that always outputs an \egal{} outcome. Then $\mechanism$ is not strategyproof, even in the \emph{UP} setting.
\end{proposition}
\begin{proof}
    Consider an instance with $P = \{p_1,p_2,p_3, p_4\}$, $n=3$, $\ell = 2$, and the  approval sets $\mathbf{S}_1, \mathbf{S}_2, \mathbf{S}_3$ such that $S_{11} = S_{21} = S_{31} = \{p_1\}$ and $S_{i2} = \{p_i\}$ for each $i \in \{1,2,3\}$.

Consider an outcome $\mathbf{o} = (o_1,o_2)$.
If $o_1\neq p_1$, at most one agent receives positive utility from $\mathbf o$, so the egalitarian welfare of $\mathbf o$ is $0$. 
In contrast, if $o_1=p_1$, the egalitarian welfare of $\mathbf o$ is at least $1$.
Thus, when agents report truthfully, $\mechanism$ outputs an outcome ${\mathbf o}^* = (o_1^*, o_2^*)$ with $o^*_1=p_1$.
 
Assume without loss of generality that $o^*_2 =  p_2$. Then $u_1({\mathbf o}^*)=1$. 
Now, suppose that agent~$1$ misreports their approval vector as $\mathbf{S}'_1 = (S'_{11}, S'_{12})$, where $S'_{11}=\{p_4\}$, 
$S'_{12}=\{p_1\}$.
In this case, the only \egal{} outcome for the reported preferences is ${\mathbf o}' = (p_1,p_1)$, so $\mechanism$ is forced to output
${\mathbf o}'$. Moreover, agent~$1$'s utility (with respect to their true preferences) 
from ${\mathbf o}'$ is $u_1(\mathbf{o}') = 2 > u_1(\mathbf{o}^*)$, 
i.e., agent $1$ has an incentive to misreport.
\end{proof}

Having ruled out compatibility of \egal{} and strategyproofness,
we consider a relaxation of strategyproofness known as \emph{non-obvious manipulability}. It was introduced by \citet{troyan2020obvious}, and
has been studied in the single-round multiwinner voting literature \cite{arribillaga2024obviousmanipulations,aziz2021obviousmanipulability}.
It is formally defined as follows.

\begin{definition}[Non-Obvious Manipulability] \label{defn-NOM}
    A mechanism $\mechanism$ is \emph{not obviously manipulable (NOM)} if for 
    every instance $(N, P, \ell, ({\mathbf S}_i)_{i\in N})$, 
    each agent $i \in N$, and each approval vector $\mathbf{B}_i$ that $i$ may report, 
    the following  conditions hold:
        \begin{align*}
        \min_{\mathcal{S}_{-i}\in \Sigma_{P, \ell}^{n-1}} u_i(\mechanism(\mathcal{S}_{-i}, \mathbf{S}_i)) &\geq  \min_{\mathcal{S}_{-i}\in\Sigma_{P, \ell}^{n-1}} u_i(\mechanism(\mathcal{S}_{-i}, \mathbf{B}_i))\\
        \max_{\mathcal{S}_{-i}\in\Sigma_{P, \ell}^{n-1}} u_i(\mechanism(\mathcal{S}_{-i}, \mathbf{S}_i)) &\geq  \max_{\mathcal{S}_{-i}\in\Sigma_{P, \ell}^{n-1}} u_i(\mechanism(\mathcal{S}_{-i}, \mathbf{B}_i)),
        \end{align*}
    where $\Sigma_{P, \ell}^{n-1}$ denotes the space of all $(n-1)$-voter profiles
    where each voter expresses her approvals of projects in $P$ over $\ell$ steps.
\end{definition}

Intuitively, a mechanism is NOM if an agent cannot increase her worst-case utility or her best-case utility (with respect to her true utility function) by misreporting.
Clearly, strategyproofness implies NOM: if a mechanism is strategyproof, 
no agent can increase her utility in \emph{any} case by misreporting.

While NOM is a much weaker condition than strategyproofness, it turns out that
it is still incompatible with \egal{}.

\begin{proposition}\label{prop:NOM_no}
    Let $\mechanism$ be a deterministic mechanism that always outputs an \egal{} outcome. Then $\mechanism$ is not NOM.
\end{proposition}
\begin{proof}
    We will prove that an agent can increase her worst-case utility, and hence the mechanism fails NOM.

    Fix $P = \{p_1,p_2\}$, $n = \ell = 2$. 
    Consider first an instance ${\mathcal I}=(N, P, \ell, ({\mathbf S}_i)_{i\in N})$
    with $S_{11} = S_{21} = \varnothing$,  $S_{12} = \{p_1\}$, and $S_{22} = \{p_2\}$.
    Let $\mathbf{o} = (o_1, o_2)$ be the output of $\mechanism$ on this instance;
    assume without loss of generality that $o_2 = p_1$. 

    Now, consider an instance ${\mathcal I}'=(N, P, \ell, ({\mathbf S}'_i)_{i\in N})$
    with $S'_{11} = \{p_1, p_2\}$, $S'_{21} = \varnothing$,  $S'_{12} = \{p_1\}$, and $S'_{22} = \{p_2\}$.
    For $\mathbf{o'}=(o_1', o_2')$ to be an \egal{} outcome for this instance, 
    it has to provide positive utility to both agents; this is only possible if $o'_2=p_2$.
    Thus, the utility of agent $1$ from the outcome selected by $\mechanism$ is $1$.
    
    However, we will now argue that if the first agent misreports her approval vector as $(\varnothing, \{p_1\})$, 
    she is guaranteed utility $2$ no matter what the second agent reports, i.e., her worst-case utility is~$2$.

    Indeed, suppose agent~$2$ reports $(\varnothing, \{p_2\})$, in which case the agents have the same approval vectors as in $\mathcal I$. We assumed that on $\mathcal I$ 
    mechanism $\mechanism$
    selects $p_1$ at timestep $2$ (and one of $p_1, p_2$ at timestep 1), so the utility of agent 1 (with respect to her true preferences in ${\mathcal I}'$) is $2$.
    On the other hand, if agent~$2$ reports $(\varnothing, \{p_1\})$ or $(\varnothing, \{p_1, p_2\})$, 
    then $\mechanism$
    selects $p_1$ at timestep $2$ (and one of $p_1, p_2$ at timestep 1), as this 
    is the only way to guarantee positive utility to both agents.
    Similarly, if $S_{21}\neq\varnothing$, there is a way to provide positive utility to both agents, by 
    selecting a project
    from $S_{21}$ at the first timestep, and $p_1$ at the second timestep. However, for agent 1 to obtain positive utility it is necessary to select $p_1$ at the second timestep, so the mechanism is forced to do so.
    That is,  if agent~$1$ reports $(\varnothing, \{p_1\})$, $\mechanism$ selects
    an outcome $(o^*_1, o^*_2)$ with $o^*_2=p_1$, and this outcome provides utility~$2$
    to agent~$1$ according to their preferences in ${\mathcal I}'$.
\end{proof}
However, we obtain a positive result for the CP setting.
Let $\mechanism_\text{lex}$ be the mechanism that outputs an \egal{}
outcome, breaking ties in favor of agents with lower indices. 
Formally, we define a partial order $\succ_\text{lex}$ on the set $\Pi({\mathcal I})$ of possible outcomes 
for a given instance $\mathcal I$ as follows: 
(1) if $\egal({\mathbf o})>\egal({\mathbf o}')$, 
then $\mathbf o\succ_\text{lex}{\mathbf o}'$;
(2) if $\egal({\mathbf o}) = \egal({\mathbf o}')$
and there is an $i\in N$ such that 
$u_{i'}({\mathbf o})=u_{i'}({\mathbf o}')$ for $i'<i$
and $u_{i}({\mathbf o})>u_{i}({\mathbf o}')$
then $\mathbf o\succ_\text{lex}{\mathbf o}'$.
Note that two outcomes are incomparable under $\succ_\text{lex}$ if and only if they provide the same utility to all agents; we say that such outcomes are {\em utility-equivalent}, 
and complete $\succ_\text{lex}$ to a total order $\succ$ on $\Pi({\mathcal I})$ arbitrarily.
$\mechanism_\text{lex}$ outputs an outcome $\mathbf o$
with $\mathbf o\succ{\mathbf o}'$ for all ${\mathbf o}'\in\Pi({\mathcal I})\setminus\{{\mathbf o}\}$.

\begin{theorem}\label{thm:NOM_yes_CP}
    $\mechanism_\text{lex}$ is \emph{NOM} in the \emph{CP} setting.
\end{theorem}
\begin{proof}
    In the CP setting, the best-case utility of each agent when they report truthfully is $\ell$: this is achieved, e.g., if all other agents have the same preferences. Thus, it remains to establish that under $\mechanism_\text{lex}$
    no agent can improve their worst-case utility by misreporting.

Let $\mathbf{S}_i$ be the true approval vector of agent $i$, and let $\mathbf{B}_i$ be another approval
vector that $i$ may report. Consider a minimum-length sequence of elementary operations that transforms $\mathbf{S}_i$
into $\mathbf{B}_i$ by first adding approvals in $B_{it}\setminus S_{it}$, $t\in [\ell]$, one by one,
and then removing approvals in $S_{it}\setminus B_{it}$, $t\in [\ell]$, one by one.
Let $\mathbf{X}^0, \mathbf{X}^1, \dots, \mathbf{X}^k,\dots, \mathbf{X}^{\gamma+1}$ be the resulting sequence
of approval vectors, with $\mathbf{X}^0=\mathbf{S}_i$, $\mathbf{X}^{\gamma+1}=\mathbf{B}_i$. Note that, since we add approvals first,   
all entries of each approval vector in this sequence are non-empty subsets of $P$, i.e., we remain in the CP setting.
Suppose this sequence starts with $k$ additions, so that 
$\mathbf{X}^{s}$ is obtained from $\mathbf{X}^{s-1}$ by adding a single approval if $s\le k$
and by deleting a single approval if $s>k$.

We will first show that, for any fixed list ${\mathcal S}_{-i}$ of other agents' approval vectors, reporting $\mathbf{X}^{k}$ instead of
$\mathbf{X}^{0}={\mathbf S}_i$ does not increase $i$'s utility; this implies that reporting $\mathbf{X}^{k}$ instead of
$\mathbf{X}^{0}$ does not increase $i$'s worst-case utility.
Then, we will show that for all $s=k+1, \dots, \gamma+1$ reporting $\mathbf{X}^{s}$ instead of
$\mathbf{X}^{s-1}$ does not increase $i$'s worst-case utility either. This implies that reporting
$\mathbf{B}_i$ instead of $\mathbf{S}_i$ does not increase $i$'s worst-case utility.

Fix a list $\mathcal{S}_{-i}$ of other agents' approval vectors, and
let $\mathcal S = (\mathcal{S}_{-i}, \mathbf{S}_i)$, 
${\mathcal S}' = (\mathcal{S}_{-i}, \mathbf{X}^k)$.
Set $\mechanism_\text{lex}({\mathcal S}) = \mathbf o$, $\mechanism_\text{lex}({\mathcal S}') = {\mathbf o}'$. 
If ${\mathbf o}'$ and $\mathbf o$ are utility-equivalent at $\mathcal S$, we are done, as this means that $i$ does not benefit from reporting ${\mathbf X}^k$ instead of ${\mathbf S_i}$; hence, assume that this is not the case.
Let $\eta$ be the egalitarian welfare of ${\mathbf o}'$ with respect to the reported utilities at ${\mathcal S}'$,
and let $\eta_i$ be the utility of $i$ at ${\mathbf o}'$ according to $\mathbf{X}^k$;
note that $\eta_i\ge \eta$. Moreover, since ${\mathbf X}^k$ is obtained from ${\mathbf X}^0$ by adding approvals, it holds that $\eta_i\ge u_i({\mathbf o}')$.
By choosing ${\mathbf o}'$ at $\mathcal{S}$, the mechanism can guarantee
utility $\eta$ to all agents other than $i$, and $u_i({\mathbf o}')\le \eta_i$ to $i$. 
If $u_i({\mathbf o}')\le \eta$, the egalitarian welfare of choosing ${\mathbf o}'$
at $\mathcal{S}$ is $u_i({\mathbf o}')$, so under
any \egal{} outcome at $\mathcal{S}$ (in particular, under the outcome chosen by $\mechanism_\text{lex}$)  the utility of $i$
is at least $u_i({\mathbf o}')$; hence, in this case we are done.

Now, suppose $u_i({\mathbf o}')> \eta$.
Then the egalitarian
welfare of ${\mathbf o}'$ at $\mathcal{S}$ is $\eta$, so
the egalitarian welfare of ${\mathbf o}$ at $\mathcal{S}$ is at least $\eta$.
Moreover, it cannot be strictly higher than $\eta$, because then 
the egalitarian welfare of ${\mathbf o}$ at ${\mathcal S}'$ according to the reported
utilities would be strictly higher than $\eta$ as well, a contradiction with
$\mechanism_\text{lex}$ outputting ${\mathbf o}'$ on $\mathcal{S}'$.
Thus, ${\mathbf o}'$ and $\mathbf o$ provide the same egalitarian welfare at $\mathcal{S}$. As we assumed that ${\mathbf o}'$ and $\mathbf o$ are not utility-equivalent, it has to be the case that $\mechanism_\text{lex}$ chooses ${\mathbf o}$ over ${\mathbf o}'$
at $\mathcal{S}$ because of condition~(2) in the definition of~$\succ_\text{lex}$.
Let $j$ be the smallest index
such that $u_j({\mathbf o}) > u_j({\mathbf o}')$. If $j\ge i$, 
we are done, as this means that $u_i({\mathbf o}) = u_i({\mathbf o}')$, 
so $i$ does not benefit from reporting $\mathbf{X}^k$ instead of $\mathbf{S}_i$.
Now, suppose that $j<i$. Note that
${\mathbf o}'$ and $\mathbf o$ provide the same egalitarian welfare at ${\mathcal S}'$, and
agents $1, \dots, j$ have the same preferences
in ${\mathcal S}'$ and ${\mathcal S}$. 
Hence, $\mechanism_\text{lex}$ should choose ${\mathbf o}$ over ${\mathbf o}'$ at ${\mathcal S}'$, 
a contradiction with the choice of~${\mathbf o}'$.
    
We will now consider $s > k$ and 
    argue that for every $\mathcal{S}_{-i}$ there is an $\mathcal{S}'_{-i}$ such that 
    $u_i(\mechanism_\text{lex}(\mathcal{S}_{-i}, \mathbf{X}^{s-1})) \geq u_i(\mechanism_\text{lex}(\mathcal{S}'_{-i}, \mathbf{X}^{s}))$. 
    Suppose that $\mathbf{X}^s$ is obtained from $\mathbf{X}^{s-1}$ by deleting a project $p$ at timestep $t$. 
    Let $\mathbf{o} = \mechanism_\text{lex}(\mathcal{S}_{-i}, \mathbf{X}^{s-1})$. 
    If $o_t\neq p$ then the outcome $\mechanism_\text{lex}(\mathcal{S}_{-i}, \mathbf{X}^{s})$ is utility-equivalent to $\mathbf o$.
    Otherwise, consider a project $p' \neq p$ approved by agent $i$ at timestep $t$ according to $\mathbf{X}^{s}$. 
    We construct $\mathcal{S}'_{-i}$ by swapping all other agents' approvals 
    for $p$ and $p'$ at timestep $t$. 
    Consider an outcome ${\mathbf o}'$ that selects $p'$ a timestep $t$ and coincides with $\mathbf o$ at all other timesteps. 
    
    For each agent $j\in N$ it holds that the utility of $j$ from $\mathbf o$ at $(\mathcal{S}_{-i}, \mathbf{X}^{s-1})$ is the same as her utility from ${\mathbf o}'$ at $(\mathcal{S}'_{-i}, \mathbf{X}^{s})$. We claim that this implies that when we execute $\mechanism_\text{lex}$ on 
    $(\mathcal{S}'_{-i}, \mathbf{X}^{s})$, it chooses an outcome that is utility-equivalent to ${\mathbf o}'$. Indeed, suppose there is an outcome $\widetilde{{\mathbf o}'}$ 
    such that at $(\mathcal{S}'_{-i}, \mathbf{X}^{s})$ it holds that $\widetilde{{\mathbf o}'}\succ_\text{lex} {\mathbf o}'$. 
    Consider at outcome $\widetilde{\mathbf o}$ that is identical to $\widetilde{{\mathbf o}'}$ except that if $\widetilde{{\mathbf o}'}$ chooses $p'$ at $t$ then $\widetilde{\mathbf o}$ chooses $p$ at $t$. Then $\widetilde{{\mathbf o}}\succ_\text{lex} {\mathbf o}$, a contradiction with our choice of $\mathbf o$.

    On the other hand, as agent $i$ approves $p$, we have $u_i(\mathbf{o})\geq u_i(\mathbf{o}')$.
    Hence,  for each $\mathcal{S}_{-i}$ there is an $\mathcal{S}'_{-i}$ such that $u_i(\mechanism_\text{lex}(\mathcal{S}_{-i}, \mathbf{X}^{s-1})) \geq u_i(\mechanism_\text{lex}(\mathcal{S}'_{-i}, \mathbf{X}^s))$, 
    i.e., reporting $\mathbf{X}^s$ instead of $\mathbf{X}^{s-1}$ does not increase $i$'s worst-case utility.
    As this holds for all $s>k$, the proof is complete.
\end{proof}


\section{Proportionality}\label{sec:prop}
Another property that may be desirable in the context
of temporal voting (and has been considered by others in similar settings 
\cite{conitzer2017fairpublic,elkind2022temporalslot})
is \emph{proportionality} (PROP).
\begin{definition}[Proportionality]  \label{defn-prop}
    Given an instance ${\mathcal I} = (N, P, \ell, ({\mathbf S}_i)_{i\in N})$, 
    for each $i\in N$ let $\mu_i = |\{t \in [\ell]: S_{it} \neq \varnothing\}|$.
    We say that an outcome $\mathbf{o}$ is \emph{proportional (PROP)} 
    for $\mathcal I$ if for all $i \in N$ 
    it holds that $u_i(\mathbf{o}) \geq \lfloor \frac{\mu_i}{n} \rfloor$.
\end{definition}

We note that proportionality is often understood as guaranteeing each agent at least
$1/n$-th of her maximum utility, which would correspond to using $\frac{\mu_i}{n}$
instead of $\lfloor \frac{\mu_i}{n} \rfloor$ in the right-hand side of our 
definition~\cite{conitzer2017fairpublic,elkind2022temporalslot,igarashi2023repeated}.
However, the requirement that $u_i(\mathbf{o}) \geq \frac{\mu_i}{n}$ may be impossible 
to satisfy: e.g. if $N=\{1, 2\}$, $\ell=3$, $P=\{p_1, p_2\}$ and for $i=1,2$
agent $i$ approves project $p_i$ at each timestep, we cannot simultaneously
guarantee utility $3/2$ to both agents. 
Moreover, the proof of Theorem~\ref{thm:npcomplete} 
can be used to show that deciding whether a given instance admits an outcome
$\mathbf o$ with $u_i(\mathbf{o}) \geq \frac{\mu_i}{n}$ for all $i\in N$
is NP-complete. 
\begin{proposition}\label{prop:prop}
    Given an instance ${\mathcal I}=(N, P, \ell, ({\mathbf S}_i)_{i\in N})$, it is {\em NP}-complete to decide whether there exists an outcome $\mathbf o$ such that $u_i({\mathbf o})\ge \mu_i/n$ for each $i\in N$. The hardness result holds even in the UP setting.
\end{proposition}
\begin{proof}
    The proof of Theorem~\ref{thm:npcomplete} proceeds by a reduction from {\sc Vertex Cover}. We can assume without loss of generality that the input instance $(G=(V, E), k)$ of {\sc Vertex Cover} in that reduction satisfies
    $|V|\ge 3$ (as otherwise the problem is trivial)
    and that $|E|>|V|$; if the latter condition does not hold, we can modify our instance by adding a $2|V|$-vertex clique that is disjoint from the rest of the graph and increasing $k$ by $2|V|-1$ (as we need $2|V|-1$ vertices to cover the clique), so that the modified graph has $3|V|$ vertices and 
    at least $2|V|(2|V|-1)/2\ge 5|V|>3|V|$ edges.

    The reduction produces an instance of \egaldec{} with $n=|E|+|V|-k$ agents and $\ell = |V|$ timesteps, where all agents have unique preferences, so that $\mu_i=\ell$ for each $i\in N$. As $|E|>|V|\ge k$,  we have $n>\ell$ and hence $0<\mu_i<n$ for each $i\in N$. Consequently, 
    an outcome $\mathbf o$ satisfies $u_i({\mathbf o})\ge \mu_i/n$ for each $i\in N$ if and only if $u_i({\mathbf o})\ge 1$ for each $i\in N$, and the proof of 
    Theorem~\ref{thm:npcomplete} shows that the latter condition is satisfied if and only if we start with a yes-instance of {\sc Vertex Cover}.
\end{proof}

In contrast, our definition of proportionality can be satisfied
by a simple greedy algorithm. This follows from a similar result 
obtained by \citet{conitzer2017fairpublic} 
in the setting of public decision-making; we provide a proof of the following theorem 
in the appendix.

\begin{theorem} \label{thm:prop}
     A \emph{PROP} outcome always exists and can be computed by a polynomial-time greedy algorithm. 
\end{theorem}

We note that PROP can be seen as a specialization of the Strong PJR axiom for temporal voting
\cite{chandak23,elkind2024verifying} to voter groups of size~$1$; this offers additional justification
for our definition. Hence, the existence
of PROP outcomes (and polynomial-time algorithms for computing them) also follows from Theorem~4.1 in the work of \citet{chandak23}. In the appendix, we discuss this connection in more detail.

As finding {\em some} PROP outcome is not hard, one may wish to select the ``best'' PROP outcome. 
A natural criterion would be to pick a PROP outcome with the maximum utilitarian or egalitarian welfare. In particular, it would be useful to have an algorithm that can determine if there exists a PROP outcome that is also \util{} or \egal{}. 

However, the proof of Proposition~\ref{prop:prop} implies that selecting a PROP outcome with maximum egalitarian welfare is computationally intractable.
Our next result shows that combining proportionality with utilitarian welfare is hard, too, 
even though both finding a PROP outcome and finding a \util{} outcome is easy.
It also implies that finding a utilitarian welfare-maximizing outcome among all PROP outcomes
is NP-hard.

\begin{theorem} \label{thm:nphard_prop_util}
     Determining if there exists a \emph{PROP} outcome that is \util{} is \emph{NP}-complete, even in the \emph{CP} setting.
\end{theorem}
\begin{proof}
    To see that this problem is is NP, recall that we can compute the maximum utilitarian welfare 
    for a given instance; thus, we can check if a given outcome is \util{} and PROP.

    By Theorem~\ref{thm:npcomplete}, in the CP setting it is NP-complete even to determine if there is an outcome $\mathbf{o}$ such that $u_i(\mathbf{o})\geq 1$ for all agents $i \in N$. 
    Given an instance ${\mathcal I}=(N, P, \ell, ({\mathbf S}_i)_{i\in N})$ with complete preferences, 
    we construct a new instance $\mathcal{I}'$ with complete preferences such that there is an outcome for $\mathcal{I}'$ that is proportional and maximizes utilitarian welfare if and only if there is an outcome for the original instance $\mathcal{I}$ that offers positive utility to all agents. 

    Note that, in the CP setting, if $\ell \geq n$, then there always exists an outcome $\mathbf{o}$ such that $u_i(\mathbf{o})\geq 1$ for all $i \in N$. Hence, assume that $\ell < n$. We construct an instance $\mathcal{I}'$ with a set of agents $N'=[2n]$, $\ell'=2n$ timesteps, and a set of projects $P'=P\cup\{q_1, \dots, q_{n+1}\}$.  
     For each $i=n+1, \dots, 2n$ and each $t\in [\ell]$
    let $P'_{it} = \{p \in P : n_{pt} \leq 2n-i\}.$ We define the approval sets $S'_{it}$ for $\mathcal{I}'$ as follows:
    \begin{equation*}
S'_{it} =
    \begin{cases}
        S_{it} & \text{if  } i\leq n \text{ and } t\leq \ell\\
        \{q_{i  - n}\} \cup P'_{it}& \text{if  } n< i\le 2n \text{ and } t\leq \ell\\
        \{q_{i}\} & \text{if  } i\leq n \text{ and } \ell <t\le 2n\\
        \{q_{n+1}\}& \text{if  } n< i\le 2n \text{ and } \ell<t\le 2n\\
                
    \end{cases}
\end{equation*}
Since all agents in $\mathcal I$ have CP preferences, this is also the case for ${\mathcal I}'$. Further, consider an arbitrary project $p\in P$. At timestep $t\in [\ell]$ this project is approved by $n_{pt}$ agents in $N$ as well as by each agent $i\in N'\setminus N$ such that $n<i\le 2n-n_{pt}$, i.e., by $n-n_{pt}$ additional agents. Thus, in total, at timestep $t\in [\ell]$ each project $p\in P$ receives $n$ approvals, whereas each project in $P'\setminus P$ receives at most one approval. On the other hand, in each of the last $2n-\ell$ timesteps $q_{n+1}$ receives $n$ approvals, whereas every other project receives at most one approval.
It follows that an outcome ${\mathbf o}'$ for $\mathcal{I}'$ maximizes \util{} if and only if $o'_t\in P$ for each $t\in [\ell]$
and $o'_t=q_{n+1}$ for each $t>\ell$.
Moreover, since in ${\mathcal I}'$ the agents have complete preferences, ${\mathbf o}'$ is proportional if and only if each agent's utility is at least $\lfloor\frac{\ell'}{2n}\rfloor=1$.

We now claim that there is an outcome ${\mathbf o}'$ that is proportional and maximizes \util{} for $\mathcal{I}'$ if and only if there is an outcome $\mathbf{o}$ for $\mathcal{I}$ such that $u_i(\mathbf{o})\geq 1$ for all agents $i \in N$. 

For the `if' direction, suppose there is an outcome $\mathbf{o}$ for $\mathcal{I}$ such that $u_i(\mathbf{o})\geq 1$ for all agents $i \in N$. We construct outcome $\mathbf{o}'$ by 
setting $o'_t=o_t$ for  $t\in [\ell]$ and $o'_t = q_{n+1}$ for $t=\ell+1, \dots, 2n$. Our characterization of the \util{} outcomes implies that ${\mathbf o}'$ maximizes the utilitarian welfare. Moreover, 
for each agent $i\in [n]$ we have $u_i(\mathbf{o}')\geq 1$, as ${\mathbf o}'$ coincides with $\mathbf{o}$ for the first $\ell$ timesteps. For agents $i=n+1, \dots,  2n$, we have $u_i(\mathbf{o}')\geq \ell$, as these agents approve the project $q_{n+1}$ for the last $2n-\ell\ge \ell$ timesteps. Hence, $\mathbf{o}'$ achieves proportionality and maximizes \util{}.

For the `only if' direction, suppose there exists an outcome $\mathbf{o}'$ for $\mathcal{I}'$ that achieves proportionality and maximizes \util{}. As $\mathbf{o}'$ achieves proportionality,  we have $u_i(\mathbf{o}')\geq 1$ for all $i\in N$. Furthermore, as $\mathbf{o}'$ maximizes \util{}, we have $o'_t = q_{n+1}$ for $t>\ell$, so no agent in $N$ derives positive utility from the last $2n-\ell$ timesteps.
 Hence, each agent $i \in N$ derives positive utility from one of the first $\ell$ timesteps. We construct the outcome $\mathbf{o}$ by setting $o_t=o'_t$ for $t\in [\ell]$. It then holds that 
  $u_i(\mathbf{o})\geq 1$ for all agents $i \in N$. 
\end{proof}

We have established that simultaneously achieving optimal welfare and proportionality is computationally hard. 
A natural challenge, then, is to quantify the impact of proportionality on welfare. To address this challenge, we use the 
concepts of the \emph{price of fairness} and the \emph{strong price of fairness} \cite{bei2021price}, which have been formulated for several proportionality guarantees in the single-round multiwinner voting literature \cite{brill2024priceable,elkind2022price,lackner2020util}.
As we focus on PROP in this paper, we instantiate the definition of the (strong) price of fairness 
accordingly, as follows.

Given a problem instance $\mathcal{I}$, let $\Pi_{\text{PROP}}(\mathcal{I})\subseteq \Pi({\mathcal I})$ denote the set of all proportional outcomes for $\mathcal{I}$. 
Furthermore, given a welfare objective $W\in\{\egal{},\util{}\}$, let $W$-OPT($\mathcal{I}$) denote the maximum $W$-welfare over all outcomes in $\Pi({\mathcal I})$.

\begin{definition}[Price of Proportionality]
    For a welfare objective $W$, the \emph{price of proportionality (PoPROP$_W$) for an instance $\mathcal I$} is the ratio between the maximum $W$-welfare of an outcome 
    for $\mathcal I$ and the maximum $W$-welfare of an outcome for $\mathcal I$ that satisfies {\em PROP}; the \emph{price of proportionality for $W$ (PoPROP$_W$)} is the supremum of 
    $\text{\em PoPROP}_W({\mathcal I})$ over all instances $\mathcal I$: 
    \begin{align*}
        \text{\emph{PoPROP}}_W(\mathcal{I}) &= \frac{W\text{\em -OPT}(\mathcal{I})}{\max_{\mathbf{o} \in \Pi_{\text{\em PROP}}(\mathcal{I})} W(\mathbf{o})}; \\
        \text{\emph{PoPROP}}_W &= \sup_{\mathcal{I}}\text{\emph{PoPROP}}_W(\mathcal{I}).
    \end{align*}
\end{definition}

\begin{definition}[Strong Price of Proportionality]
    For a welfare objective $W$, the \emph{strong price of proportionality (s-PoPROP$_W$) for an instance $\mathcal I$} is 
    the ratio between the maximum $W$-welfare of an outcome 
    for $\mathcal I$ and the minimum $W$-welfare of an outcome for $\mathcal I$ that satisfies {\em PROP}; the \emph{strong price of proportionality for $W$ (s-PoPROP$_W$)} is 
    the supremum of $\text{\em s-PoPROP}({\mathcal I})$ over all instances $\mathcal I$:
    \begin{align*}
        \text{\emph{s-PoPROP}}_W({\mathcal I}) &= \frac{W\text{\em -OPT}(\mathcal{I})}{\min_{\mathbf{o} \in \Pi_{\text{\em PROP}}(\mathcal{I})} W(\mathbf{o})}; \\
        \text{\emph{s-PoPROP}}_W &= \sup_{\mathcal{I}}\text{\emph{s-PoPROP}}_W({\mathcal I}).
    \end{align*}
\end{definition}

We first observe that requiring proportionality has no impact on egalitarian welfare:
any outcome $\mathbf o$ can be transformed (in polynomial time) 
into a proportional outcome ${\mathbf o}'$ 
so that the egalitarian welfare of ${\mathbf o}'$ is at least as high as that of $\mathbf o$.

\begin{proposition} \label{prop:prop_egal}
    Given an outcome $\mathbf o$, we can construct in polynomial time 
    another outcome ${\mathbf o}'$ such that ${\mathbf o}'$ is proportional and  $\egal({\mathbf o}')\ge \egal({\mathbf o})$.
\end{proposition}

\begin{proof}
    Consider an outcome $\mathbf{o}$. If $\mathbf o$ satisfies PROP, we are done, so assume that this is not the case. Let $\lambda=\egal({\mathbf o})$.
    Reorder the agents in $N$ so that $\mu_i\le \mu_{i'}$ for 
    $1\le i<i'\le n$. Since $\mathbf o$ fails PROP, there exists an $i\in N$ such that 
    $u_i(\mathbf{o})< \lfloor\frac{\mu_i}{n}\rfloor$; 
    let $q$ be the smallest value of $i$ for which this is the case. 

    We will construct an outcome ${\mathbf o}'$ that is proportional and whose egalitarian welfare is at least as high as that of $\mathbf o$, i.e., ${\mathbf o}'$ satisfies
    $u_i({\mathbf o}')\ge \max\{\lambda, \lfloor\frac{\mu_i}{n}\rfloor\}$ for all $i\in N$.
    To this end, we proceed in two stages.

    Note first that for each agent $i\in [q-1]$ we have $u_i({\mathbf o})\ge \lfloor\frac{\mu_i}{n}\rfloor$ (by our choice of $q$) and $u_i({\mathbf o})\ge \lambda$ 
    (by the definition of $\lambda$). Thus, for each $i\in [q-1]$ there are
    at least $\max\{\lambda, \lfloor\frac{\mu_i}{n}\rfloor\}$ timesteps $t\in [\ell]$ 
    in which $i$ approves $o_t$. During the first stage, we ask each agent $i\in [q-1]$
    to mark $\max\{\lambda, \lfloor\frac{\mu_i}{n}\rfloor\}$ timesteps $t\in [\ell]$
    with $o_t\in S_{it}$; multiple agents are allowed to mark the same timestep.
    We have $\lambda \le u_q({\mathbf o}) < \lfloor\frac{\mu_q}{n}\rfloor$ 
    (by our choice of $\lambda$ and $q$)
    and $\lfloor\frac{\mu_i}{n}\rfloor\le \lfloor\frac{\mu_q}{n}\rfloor$ for $i<q$, so
    the total number of timesteps marked during the first stage does not exceed 
    $(q-1)\lfloor\frac{\mu_q}{n}\rfloor$.
    Then, for each marked timestep $t$ we set $o'_t=o_t$.

    During the second stage, we proceed greedily, just as in the proof of Theorem~\ref{thm:prop}. 
    That is, for each $i=q, \dots, n$ we ask agent $i$ to mark $\lfloor\frac{\mu_i}{n}\rfloor$ 
    previously unmarked timesteps $t$
    with $S_{it}\neq\varnothing$; for each such timestep $t$ 
    we pick a project $p\in S_{it}$ and set $o'_t=p$.
    To see why each agent $i$ with $i\ge q$ can find $\lfloor\frac{\mu_i}{n}\rfloor$ suitable timesteps, 
    note that her predecessors have marked at most
    $$
    (q-1)\left\lfloor\frac{\mu_q}{n}\right\rfloor + \sum_{i'=q}^{i-1} \left\lfloor\frac{\mu_{i'}}{n}\right\rfloor \le (i-1)\left\lfloor\frac{\mu_i}{n}\right\rfloor \le \mu_i - \left\lfloor\frac{\mu_i}{n}\right\rfloor
    $$
    timesteps, so at least $\lfloor\frac{\mu_i}{n}\rfloor$ of the $\mu_i$ timesteps 
    in which agent $i$ approves some projects remain available to agent $i$.

    By construction, for each $i<q$ we have 
    $u_i({\mathbf o}')\ge \max\{\lambda, \lfloor\frac{\mu_i}{n}\rfloor\}$, and for $i\ge q$ 
    we have $u_i({\mathbf o}')\ge \lfloor\frac{\mu_i}{n}\rfloor\ge \lfloor\frac{\mu_q}{n}\rfloor> u_q({\mathbf o})\ge \lambda$.
    Hence, ${\mathbf o}'$ satisfies PROP and guarantees utility at least $\lambda$ to all agents.
\end{proof}
By applying Proposition~\ref{prop:prop_egal} to an outcome $\mathbf o$
that maximizes the egalitarian welfare, we obtain the following corollary.
\begin{corollary} \label{cor:prop_egal}
    $\text{\emph{PoPROP}}_\egal{}=1$.
\end{corollary}

In contrast, for utilitarian welfare, the price of proportionality scales as $\sqrt{n}$, 
even in the CP setting.

\begin{theorem}
    In the \emph{CP} setting, $\text{\emph{PoPROP}}_\util = \frac{n}{2 \sqrt {n} - 1}$ and hence $\Theta(\sqrt{n})$. The lower bound applies even in the {\em UP} setting.
\end{theorem}
\begin{proof}
    We first establish the lower bound.
    Given an integer $k\ge 2$,
    we construct an instance $\mathcal I$ with 
    $n=k^2$ agents, a set of projects $P = \{p_0, p_1,\dots,p_{n-k}\}$, and $\ell = n$.
    The agents have static preferences: 
    their approval sets for all timesteps $t \in [\ell]$ are defined by
    \begin{equation*}
    S_{it} =
    \begin{cases}
        \{p_i\} & \text{if  } i \leq n - k\\
        \{p_0\} & \text{otherwise. } 
    \end{cases}
\end{equation*}
Note that in this instance agents have unique preferences. The (unique) \util{} outcome chooses project $p_0$ for all timesteps. As $k$ agents approve $p_0$, the utilitarian welfare of this outcome is $nk$. 
However, in order for an outcome to be proportional, each project must be selected at least once. 
Hence, for each ${\mathbf o}\in \Pi_{\text{PROP}}(\mathcal{I})$
we have $\util({\mathbf o})\le n - k + k^2$. 
Then, we obtain
\begin{equation*}
    \text{PoPROP}_\util(\mathcal I) \ge \frac{nk} {n - k + k^2} = \frac{n \cdot \sqrt{n}}{2 n - \sqrt{n}} = \frac{n}{2 \sqrt {n} - 1}.
\end{equation*}

To establish the upper bound, given an instance $\mathcal I$, we explicitly construct an outcome $\mathbf o$ such that $\util({\mathbf o})$ is at least  a $(\frac{2}{\sqrt {n}} - \frac{1}{n})$-fraction of the maximum utilitarian welfare for $\mathcal I$. 

We first focus on the case $\ell=n$; later, we will explain how to extend our analysis to other values of $\ell$. As we consider the CP setting, in this case proportionality means that for every agent $i\in N$ there should be at least one timestep where $i$ approves the selected project. 

Recall that 
$n_{pt}$ denotes the number of agents approving project $p$ at timestep $t$, and let
$k=\max_{(p, t)\in P\times [\ell]}n_{pt}$.
We construct an outcome $\mathbf o$ in two stages.

In the first stage, we start by setting $T_1=[\ell]$. Then, 
for $j=1, \dots, k$ we
pick a pair $(p_j, t_j)\in\arg\max_{(p, t)\in P\times T_j}n_{pt}$, and set $o_{t_j}=p_j$ and $T_{j+1}=T_{j}\setminus \{t_{j}\}$.
This concludes the first stage. 

Note that during the first stage we select projects for $k\ge 1$ timesteps, so $|T_{k+1}|=n-k$ timesteps remain available.
Further, at least $k=n_{p_1t_1}$ agents approve $p_1 = o_{t_1}$ at timestep $t_1$, so there are at most $n-k$ agents who obtain utility $0$ from the partial outcome constructed in the first stage. During the second stage, we allocate each such agent $i$ a distinct timestep $t\in T_{k+1}$, and set $o_t$ to a project in $S_{it}$ (recall that the agents have complete preferences, so $S_{it}\neq\varnothing$); this is feasible as we have at most $n-k$ agents unsatisfied by stage 1, and $n-k$ available timesteps. 

By construction, each agent's utility from $\mathbf o$ is at least $1$, so $\mathbf o$ is proportional. 
Furthermore, for each $j\in [k]$ project $p_j$ is among the projects that receive the highest number of approvals at timestep $t_j$, i.e., our selections during the first stage maximize the utilitarian welfare for timesteps $t_1 \dots, t_k$. 
Let $U=\sum_{j\in [k]}n_{p_jt_j}$ be the utilitarian welfare obtained from these timesteps.

We are now ready to bound $\util({\mathbf o})$ and the maximum utilitarian welfare in our instance.
First, note that $\util({\mathbf o})\ge U+(n-k)$:
we obtain $U$ from the first stage, and each timestep assigned in the second stage contributes at least $1$.
Further, let $s=n_{p_kt_k}$, i.e., $s$ is the number of approvals received by project $p_k$ in timestep $t_k$. Since during the first stage we sequentially selected pairs $(p, t)$ so as to maximize $n_{pt}$, we have $U\ge sk$. Moreover, for each timestep $t\in T_{k+1}$ (i.e., a timestep for which we did not assign a project during the first stage) we have 
$\max_{p\in P}n_{pt}\le s$, so these timesteps can contribute at most $s(n-k)$ to the utilitarian welfare, no matter which projects we select.
Thus, the maximum utilitarian welfare for our instance is at most $U+s(n-k)$. Therefore, we have
$$
\text{PoPROP}_\util({\mathcal I}) \leq \frac{U+s(n-k)}{U + (n-k)}.
$$
Since $U\geq sk$, $s\ge 1$, and the function $\frac{x+A}{x+B}$ is monotonically non-increasing for $A\ge B$, 
we obtain
$$
\text{PoPROP}_\util({\mathcal I}) \leq \frac{sk+s(n-k)}{sk + (n-k)}  = \frac{sn}{k(s-1)+n}.
$$
Further, since $s=n_{p_kt_k}\le n_{p_1t_1}=k$ and $s-1\ge 0$, we have
$$
\text{PoPROP}_\util({\mathcal I}) \leq \frac{sn}{k(s-1)+n}\le \frac{sn}{s(s-1)+n}.
$$
Applying the AM-GM inequality to $s^2$ and $n$, we obtain $s^2+n\ge 2s\sqrt{n}$, so
$$
\text{PoPROP}_\util({\mathcal I}) \le \frac{sn}{s(s-1)+n}\le \frac{sn}{2s\sqrt{n}-s}= \frac{n}{2\sqrt{n}-1}.
$$
Thus, we have established the desired upper bound on $\text{PoPROP}_\util$ for the case $n=\ell$.
Now, if $\ell<n$, then $\frac{\mu_i}{n}<1$ for all $i\in N$, so no outcome
violates the proportionality constraint and therefore $\text{PoPROP}_\util{}=1$ in that case.

Finally, if $\ell>n$, 
let $q=\lfloor \frac{\ell}{n} \rfloor$
and let $r  = \ell-q\cdot \lfloor \frac{\ell}{n} \rfloor$.
We split the $\ell$ timesteps to create $q$ groups 
$T^1, \dots, T^q$ of $n$ timesteps each; if $r>0$, we create an extra group $T^{q+1}$ of size $r$. 
For each $T^j$, $j\in [q]$, we use the construction for the case $\ell=n$
to create a partial outcome ${\mathbf o}^j=(o_t)_{t\in T^j}$.
For timesteps in $T^{q+1}$, we construct ${\mathbf o}^{q+1}$ 
by choosing projects that receive the highest number of approvals at each timestep. We then merge ${\mathbf o}^j$, $j\in [q+1]$, into a single outcome~$\mathbf o$.

As each agent obtains positive utility from each ${\mathbf o}^j$, $j\in [q]$, the utility of each agent is at least $q=\lfloor \frac{\ell}{n} \rfloor$, so PROP is satisfied. Now, for each $j\in [q+1]$ let $\alpha_j$ be the maximum utilitarian welfare achievable for timesteps in $T^j$, and let $\beta_j=\util({\mathbf o}^j)$; if $T^{q+1}=\varnothing$, we set $\alpha_{q+1}=\beta_{q+1}=0$, with the convention that $\frac00=1$. We have argued that $\text{PoPROP}_\util\le \frac{n}{2\sqrt{n}-1}$ 
as long as $\ell=n$, so we have $\frac{\alpha_j}{\beta_j}\le \frac{n}{2\sqrt{n}-1}$ for each $j\in [q]$; also, $\frac{\alpha_{q+1}}{\beta_{q+1}}=1$ by construction of ${\mathbf o}^{q+1}$.
It remains to observe that 
$$
\text{PoPROP}_\util({\mathcal I}) = \frac{\alpha_1+\dots+\alpha_{q+1}}{\beta_1+\dots+\beta_{q+1}}
\le \max_{j\in [q+1]} \frac{\alpha_j}{\beta_j}\le \frac{n}{2\sqrt{n}-1}
$$
and
$$
\frac{n}{2\sqrt{n}}\le \frac{n}{2\sqrt{n}-1}\le \frac{n}{\sqrt{n}}, 
$$
which implies $\text{PoPROP}_\util=\Theta(\sqrt{n})$.
\end{proof}

To obtain bounds on s-PoPROP, 
we first prove a technical lemma.
\begin{lemma}\label{lem:kn}
   For every $k, n\in {\mathbb N}$ such that $k\ge n$ it holds that $\lfloor \frac{k}{n} \rfloor\ge \frac{k}{2n-1}$. 
\end{lemma}
\begin{proof}
    We proceed by case analysis.
    \begin{itemize}
        \item $n=1$. 
        In this case we have 
        $\lfloor \frac{k}{n} \rfloor=k = \frac{k}{2n-1}$, so the statement of the lemma is true.
        \item $n\ge 2$, $n\le k \le 2n-1$.
        In this case we have 
        $\lfloor \frac{k}{n} \rfloor=1$ and $\frac{k}{2n-1}\le 1$, so the statement of the lemma is true as well.
        \item $n\ge 2$, $2n\le k \le 3n-1$.
     We have $3n-1\le 4n-2$, so
        $\frac{k}{2n-1}\le 2=\lfloor \frac{k}{n} \rfloor$, which establishes the statement of the lemma for this case.
        \item $n\ge 2$, $k\ge 3n$.
        Since $n\ge 2$, we have $2n-1\ge \frac{3n}{2}$ and hence $\frac{k}{2n-1}\le \frac{2k}{3n}$. Moreover, 
        $k\ge 3n$ implies $\frac{k}{3n}\ge 1$.
        Thus, we have
        $$
        \left\lfloor \frac{k}{n} \right\rfloor\ge \frac{k}{n}-1\ge \frac{k}{n}-\frac{k}{3n}=\frac{2k}{3n}\ge \frac{k}{2n-1}, 
        $$
        establishing the statement of the lemma for this case as well.
    \end{itemize}
\end{proof}

We are now ready to bound the strong price of proportionality. 
\begin{theorem}\label{thm: strong price}
     We have $\text{\em s-PoPROP}_\util=\text{\em s-PoPROP}_\egal=+\infty$. 
     However, if for all agents $i \in N$ it holds that $\lfloor \frac{\mu_i}{n} \rfloor \geq 1$, then 
     $\text{\em s-PoPROP}_\util=\text{\em s-PoPROP}_\egal=2n-1$.
\end{theorem}
\begin{proof}
    We will first argue that for general preferences the strong price of proportionality is unbounded, both for \util{} and for \egal{}. 
Consider an instance with $P = \{p,q\}$, $n\ge 2$ agents, and $\ell = 1$. 
For each $i \in N$, let $S_{i1} = \{p\}$.
Then, both possible outcomes $\mathbf{o} = (p)$ and 
$\mathbf{o}' = (q)$ satisfy PROP,
but only $\mathbf{o} = (p)$ delivers positive utility.
Hence, $\text{s-PoPROP}_\util=\text{s-PoPROP}_\egal=+\infty$.  

Next, suppose that
$\lfloor \frac{\mu_i}{n} \rfloor \geq 1$ 
for all agents $i \in N$. 
We will show that
$\text{s-PoPROP}_\util=\text{s-PoPROP}_\egal=2n-1$.

For the upper bound, consider a PROP outcome $\mathbf{o}$. 
For each agent $i \in N$ we have $u_i(\mathbf{o}) \geq \lfloor \frac{\mu_i}{n} \rfloor$.  Since $\lfloor \frac{\mu_i}{n} \rfloor \geq 1$ implies $\mu_i\ge n$, we can apply
Lemma~\ref{lem:kn} to $\mu_i$ and $n$ to obtain
$\lfloor \frac{\mu_i}{n} \rfloor\ge \frac{\mu_i}{2n-1}$, so  
\begin{align*}
\util{}({\mathbf o}) &=\sum_{i\in N}u_i(\mathbf{o})\ge \frac{1}{2n-1}\sum_{i\in N}\mu_i, \\
\egal{}({\mathbf o}) &=\min_{i\in N}u_i(\mathbf{o})\ge \frac{1}{2n-1}\min_{i\in N}\mu_i.
\end{align*}
Since the maximum utilitarian and egalitarian welfare do not exceed
     $\sum_{i\in N}\mu_i$ and $\min_{i\in N}\mu_i$, respectively, this establishes our upper bound.
     
To prove the lower bound, consider an instance with a set of $n$ agents $N$, $P = \{p,q\}$, and $\ell = 2n-1$. 
For each $i \in N$ and $t \in [\ell]$, let $S_{it} = \{p\}$. For this instance 
the maximum utilitarian welfare and 
the maximum egalitarian welfare are, respectively, $n \cdot (2n-1)$ and $2n-1$, 
obtained by selecting $p$ at every timestep. On the other hand, we have $1\le \frac{\mu_i}{n}<2$ for all $i\in N$, so ${\mathbf o}'$ with $o_1=p$, 
$o_t=q$ for $t=2, \dots, \ell$ satisfies PROP. This implies a lower bound of $2n - 1$ 
for both $\text{s-PoPROP}_\util$ and $\text{s-PoPROP}_\egal$.
\end{proof}


\section{Extensions}

So far, we have focused on \util{} and \egal{}; however, some of our results extend to the entire space of welfare measures between these two extremes, namely, the family of \emph{$p$-mean welfare} objectives. Formally, for each $p\in\mathbb R$, the {\em $p$-mean welfare} provided by an outcome $\mathbf o$ is given by $\left( \frac{1}{n} \sum_{i \in N} u_i(\mathbf{o})^p \right)^{1/p}$. 
The associated family of decision problems is defined as follows.
\begin{tcolorbox}
\textsc{$p$-mean Welfare}:\\
\textbf{Input}: An instance ${\mathcal I} = (N,P,\ell, (\mathbf{S}_i)_{i\in N})$ 
and a parameter $\lambda \in \mathbb{Z}^+$.\\
\textbf{Question}: Is there an outcome $\mathbf{o}$ that satisfies
$\left( \frac{1}{n} \sum_{i \in N} u_i(\mathbf{o})^p \right)^{1/p} \ge\lambda$?
\end{tcolorbox}

Note that setting $p = 1$ (respectively, $p= -\infty$) corresponds to the utilitarian (respectively, egalitarian) welfare. Setting $p \to 0$ corresponds to the geometric mean, or Nash welfare, which we denote by \nash{} (with the corresponding decision problem denoted by \nashdec{}). 

It can be verified that many of the computational hardness and impossibility results for \egal{} directly translate to similar results for \nash{}: 
\nash{} is obviously manipulable in the general setting and not strategyproof even in the CP setting. Moreover, \nashdec{} is NP-complete even when $m = 2$ and is W[2]-hard with respect to $\ell$. 

Regarding positive results, the XP algorithm with respect to $\ell$ is based on enumerating all possible outcomes, and  hence it works for all $p$-mean welfare measures. While our FPT algorithm (with respect to $n$) 
 for \egal{} relies on an ILP that does not extend to other welfare measures due to their non-linearity, a randomized XP algorithm has been recently proposed for a more demanding setting \cite{elkind2022temporalslot}. 
In our model, we can show that there is a deterministic XP algorithm (with respect to $n$) for any \emph{p-mean welfare} objective.
\begin{theorem}\label{thm:xp-det}
    There exists a deterministic \emph{XP} algorithm (with respect to $n$) that maximizes the $p$-mean welfare.
\end{theorem}
\begin{proof}
    We describe a dynamic programming algorithm that runs in time $\mathcal{O}((\ell+1)^{n+1}\cdot mn)$.

    Fix an instance ${\mathcal I}=(N, P, \ell, ({\mathbf S}_i)_{i\in N})$, and let $L=\{0, \dots,\ell\}$.
    Given an outcome $\mathbf o$ for $\mathcal I$, a timestep $t\in L$, and an agent $i\in N$, we write $u_i^t({\mathbf o})$
    to denote the utility obtained by agent $i$ from the first $t$ timesteps under $\mathbf o$: 
    $u_i^t({\mathbf o})=|\{j\in [t]: o_j\in S_{ij}\}|$. We collect these utilities in a vector $U^t({\mathbf o})=(u_i^t({\mathbf o}))_{i\in N}$. Note that for each $t\in L$ and each ${\mathbf o}\in\Pi({\mathcal I})$ we have $U^t({\mathbf o})\in L^n$.

    For each vector $V\in L^n$ and each $t\in L$
    let $Q[V, t]=1$ if there exists an outcome ${\mathbf o}\in\Pi({\mathcal I})$ such that $U^t({\mathbf o})=V$, 
    and let $Q[V, t]=0$ otherwise. We will now explain how to 
    compute the quantities $(Q[V, t])_{V\in L^n, t\in L}$ using dynamic programming.
    
    For the base case $t=0$, we have
\begin{equation*}
    Q[V, 0] =
    \begin{cases}
        1 & \text{if  } V = (0, \dots, 0)\\
        0 & \text{otherwise.} \\
    \end{cases}
\end{equation*}

Now, suppose $t\ge 1$.
Then for each $V\in L^n$ 
we have $Q[V, t]=1$ if and only if there is a vector $Y\in L^n$
with $Q[Y, t-1]=1$ and a project $p$ such that for each 
$i\in N$ we have either (i) $p\in S_{it}$ and $Y_i+1=V_i$ or (ii) $p\not\in S_{it}$ and $Y_i=V_i$. 
Thus, given the values of $Q[Y, t-1]$ for all $Y\in L^n$,
we can compute the values $Q[V, t]$ for all $V\in L^n$;
each value can be determined by iterating through all projects and all agents, i.e., in time ${\mathcal O}(mn)$.

As our dynamic programming table has $(\ell+1)^{n+1}$
entries, it can be filled in time ${\mathcal O}((\ell+1)^{n+1}mn)$. Once it has been filled, we can scan through all vectors
$V\in L^n$ such that $Q[V, \ell]=1$ and find one with the highest $p$-mean welfare. An outcome $\mathbf o$ that provides this welfare can be found using standard dynamic programming techniques.
\end{proof}

However, it remains open whether 
there is an FPT algorithm with respect to $n$ that maximizes \nash{} (or other $p$-mean welfare objectives with $p\neq 0, -\infty$).

\section{Conclusions}

We investigated the problem of maximizing utilitarian and egalitarian welfare for temporal elections. 
We showed that, while \util{} outcomes can be computed in polynomial time and can be achieved in a strategyproof manner, \egal{} is NP-complete and obviously manipulable. 
To circumvent the NP-hardness of \egal{}, we analyzed its parameterized complexity with respect to $n,m$ and $\ell$, and provided an approximation algorithm that is based on randomized rounding. 
We also established the existence of a NOM mechanism for \egal{} under a mild constraint on agents' preferences. 
Finally, we considered proportionality and showed that it is computationally hard to select the `best' proportional outcome. We also gave upper and lower bounds on the (strong) price of proportionality with respect to both \util{} and \egal{}. 

We note that our upper bound on the price of proportionality with respect to \util{} only applies to the CP setting. It is easy to see that $\text{PoPROP}_\util$ does not exceed $n$ even for general preferences: if there are $\ell'$ timesteps $t\in [\ell]$ with $\cup_{i\in N}S_{it}\neq\varnothing$, the utilitarian welfare does not exceed $\ell' n$, whereas selecting $o_t$ from $\cup_{i\in N}S_{it}$ for each $t\in [\ell]$ results in an outcome $\mathbf o$ that satisfies $\util({\mathbf o})\ge \ell'$.
However, obtaining non-trivial upper bounds on $\text{PoPROP}_\util$ in the general setting remains an open problem.

Some of our results extend to $p$-mean objectives other than \egal{} and \util{}. However, it remains open which of these objectives admit polynomial-time algorithms.

In our model, agents are assumed to have approval preferences, which can be thought of as  \emph{binary utilities} over projects. One can extend this model to arbitrary {\em cardinal preferences}, 
by allowing each agent $i \in N$ to have a valuation function $v_i:P\times[\ell] \rightarrow \mathbb{R}_{[0,1]}$ instead of an approval set.
Cardinal preferences have recently been studied in various social choice settings \cite{conitzer2017fairpublic,elkind2022portioning,elkind2024temporalfairdivision,fain2018publicgoods,freeman2017dynamicsocialchoice}. 
It would be interesting to investigate whether the positive results for our model extend to the model that allows for arbitrary cardinal preferences.
\begin{ack}
    Edith Elkind was supported by the AI Programme of The Alan Turing Institute and an EPSRC Grant EP/X038548/1.
\end{ack}

\bibliography{abb,mybibfile}

\newpage

\appendix

\begin{center}
\Large
\textbf{Appendix}
\end{center}

\vspace{2mm}
    
\section{Proof of Theorem \ref{thm:prop}}\label{sec:thmprop}
We describe a polynomial-time greedy algorithm for obtaining a PROP outcome.

We start by reordering the agents so that $\mu_1\le \mu_2\le\dots\le \mu_n$,
and let $T=[\ell]$. We then process the agents in $N$ one by one. For each $i\in N$
we pick a subset of timesteps $T_i\subseteq T$
with $|T_i| = \lfloor\frac{\mu_i}{n}\rfloor$ and $S_{it}\neq\varnothing$ for each $t\in T_i$. Then for each $t\in T_i$ we set $o_t=\{p\}$ for some $p\in S_{it}$; we then remove
$T_i$ from $T$. If $T$ is still non-empty after all agents have been processed, 
we select the projects for timesteps in $T$ arbitrarily; e.g., for each $t\in T$ we can select a project that receives the maximum number of approvals at $t$.

Clearly, if we manage to find an appropriate set of timesteps $T_i$ for each $i\in N$, 
we obtain a proportional outcome. Thus, it remains to argue that such a set of timesteps
can always be found.
For an agent $i$, the number of projects selected before her turn is 
\begin{align*}
    \sum_{j \in [i-1]} \Bigl\lfloor \frac{\mu_j}{n} \Bigr\rfloor 
    &\leq \sum_{j \in [i-1]} \Bigl\lfloor \frac{\mu_i}{n} \Bigr\rfloor 
    = (i -1)\Bigl\lfloor \frac{\mu_i}{n} \Bigr\rfloor\\ 
    &\leq (n -1)\Bigl\lfloor \frac{\mu_i}{n} \Bigr\rfloor 
    \leq \frac{(n -1) \mu_i}{n}.
\end{align*}
As $\mu_i - \frac{(n -1) \mu_i}{n} = \frac{\mu_i}{n} \geq \lfloor \frac{\mu_i}{n} \rfloor$,
it is indeed the case that we can find a set $T_i$ of appropriate size.
\qed

\section{Proportionality and Strong PJR}\label{sec:spjr}
The notion of proportional justified representation (PJR) was proposed by \citet{sanchez2017proportional} 
for multiwinner voting with approval preferences; it was extended to the temporal setting by \citet{bulteau2021jrperpetual} 
and \citet{chandak23}. In what follows, we use the terminology of \citet{chandak23}.

\begin{definition}\label{def:spjr}
Given an instance ${\mathcal I}=(N, P, \ell, ({\mathbf S}_i)_{i\in N})$, 
we say that a group of voters $N'$ {\em agrees in $k$ timesteps} if there exists
a subset $T'\subseteq [\ell]$ with $|T'|=k$ such that $\cap_{i\in N'}S_{it}\neq\varnothing$
for each $t\in T'$. We say that an outcome
$\mathbf o=(o_1, \dots, o_\ell)$ provides {\em strong proportional justified representation} for 
$\mathcal I$ if for every $s, k\in\mathbb N$ and every group of voters $N^*$ with 
$|N^*|\ge s\cdot\frac{n}{k}$ that agrees in $k$ timesteps there exists 
a subset of timesteps $T^*\subseteq [\ell]$ with $|T^*|=s$ such that for each $t\in T^*$
it holds that $o_t\in \cup_{i\in N^*} S_{it}$. 
\end{definition}
Now, we specialize this definition to a single voter $i$. Note that the set $\{i\}$
agrees in $\mu_i$ timesteps and $|\{i\}|=1 = \frac{\mu_i}{n}\cdot \frac{n}{\mu_i}$.
Hence, by setting $k=\mu_i$, $s=\lfloor\frac{\mu_i}{n}\rfloor$ we conclude that 
if an outcome $\mathbf o$ provides strong proportional justified representation for $\mathcal I$, 
it has to be the case that $u_i({\mathbf o})\ge \lfloor\frac{\mu_i}{n}\rfloor$, i.e., 
$\mathbf o$ is a PROP outcome.
\end{document}